\documentclass[11pt]{article}
\usepackage{epsf}
\usepackage{amsmath}
\usepackage{epsfig}
\usepackage{times}
\usepackage{amssymb}
\usepackage{amsthm}
\usepackage{setspace}
\usepackage{cite}

\usepackage{algorithmic}  
\usepackage{algorithm}

\usepackage{shadow}
\usepackage{fancybox}
\usepackage{fancyhdr}

\def\w{{\bf w}}

\def\y{{\bf y}}

\def\x{{\bf x}}

\def\x{{\mathbf x}}

\def\w{{\bf w}}

\def\x{{\bf x}}
\def\y{{\bf y}}
\def\z{{\bf z}}

\def\a{{\bf a}}

\def\h{{\bf h}}

\def\be{\begin{equation}}
\def\ee{\end{equation}}
\def\ba{\left[\begin{array}}
\def\ea{\end{array}\right]}

\def\t{{\bf t}}

\def\w{{\bf w}}

\def\x{{\bf x}}
\def\y{{\bf y}}
\def\z{{\bf z}}

\def\a{{\bf a}}

\def\1{{\bf 1}}

\def\g{{\bf g}}
\def\0{{\bf 0}}

\def\hatx{{\hat{\x}}}

\def\erfinv{\mbox{erfinv}}
\def\htheta{\hat{\theta}}

\def\cweak{c_{w}}

\def\cweaknon{c_{w}^+}
\def\betawnon{\beta_{w}^{+}}

\def\hthetawnon{\hat{\theta}_{w}^+}
\def\thetawnon{\theta_{w}^+}

\def\tauw{\tau^{+}}

\newtheorem{theorem}{Theorem}

\newtheorem{lemma}{Lemma}

\begin{document}

\begin{singlespace}

\title {Upper-bounding $\ell_1$-optimization weak thresholds
}
\author{
\textsc{Mihailo Stojnic}
\\
\\
{School of Industrial Engineering}\\
{Purdue University, West Lafayette, IN 47907} \\
{e-mail: {\tt mstojnic@purdue.edu}} }
\date{}
\maketitle

\centerline{{\bf Abstract}} \vspace*{0.1in}

In our recent work \cite{StojnicCSetam09} we considered solving under-determined systems of linear equations with sparse solutions.
In a large dimensional and
statistical context we proved that if the number of equations in the system is proportional
to the length of the unknown vector then there is a sparsity (number
of non-zero elements of the unknown vector) also proportional to the
length of the unknown vector such that a polynomial $\ell_1$-optimization technique
succeeds in solving the system. We provided lower bounds on the proportionality constants that are in a solid
numerical agreement with what one can observe through numerical experiments. Here we create a mechanism that can be used to derive the upper bounds on the proportionality constants. Moreover, the upper bounds obtained through such a mechanism match the lower bounds from \cite{StojnicCSetam09} and ultimately make the latter ones optimal.

\vspace*{0.25in} \noindent {\bf Index Terms: Linear systems of equations;
$\ell_1$-optimization; compressed sensing} .

\end{singlespace}

\section{Introduction}
\label{sec:back}

We start by defining what the problem of interest will be (the same of course was the problem of interest in \cite{StojnicCSetam09}). We will be interested in finding sparse solutions of under-determined systems of linear equations. In a more precise mathematical language we would like to find a $k$-sparse $\x$ such
that
\begin{equation}
A\x=\y \label{eq:system}
\end{equation}
where $A$ is an $m\times n$ ($m<n$) matrix and $\y$ is
an $m\times 1$ vector (see Figure
\ref{fig:model}; here and in the rest of the paper, under $k$-sparse vector we assume a vector that has at most $k$ nonzero
components). Of course, the assumption will be that such an $\x$ exists.

To make writing in the rest of the paper easier, we will assume the
so-called \emph{linear} regime, i.e. we will assume that $k=\beta n$
and that the number of equations is $m=\alpha n$ where
$\alpha$ and $\beta$ are constants independent of $n$ (more
on the non-linear regime, i.e. on the regime when $m$ is larger than
linearly proportional to $k$ can be found in e.g.
\cite{CoMu05,GiStTrVe06,GiStTrVe07}).
\begin{figure}[htb]
\centering
\centerline{\epsfig{figure=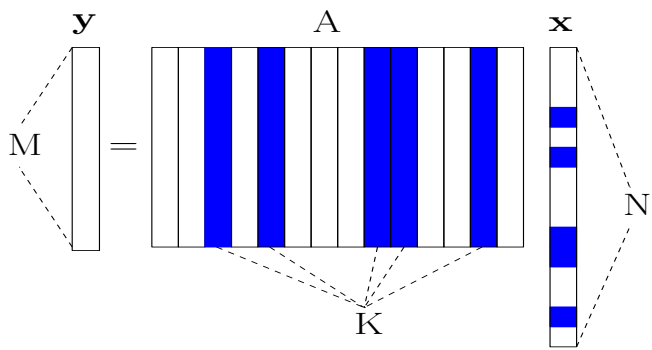,width=9cm,height=4.5cm}}
\caption{Model of a linear system; vector $\x$ is $k$-sparse}
\label{fig:model}
\end{figure}

We generally distinguish two classes of possible algorithms that can be
developed for solving (\ref{eq:system}). The first class of algorithms assumes freedom in designing matrix $A$.
If one has the freedom to design matrix $A$ then the results from \cite{FHicassp,Tarokh,MaVe05} demonstrated that the techniques from
coding theory (based on the coding/decoding of Reed-Solomon codes)
can be employed to determine \emph{any} $k$-sparse $\x$ in
(\ref{eq:system}) for any $0<\alpha\leq 1$ and any
$\beta\leq\frac{\alpha}{2}$ in polynomial time. It is relatively easy to show that under the unique recoverability assumption
$\beta$ can not be greater than $\frac{\alpha}{2}$. Therefore, as long as one is concerned with the unique recovery of
$k$-sparse $\x$ in (\ref{eq:system}) in polynomial time the results from \cite{FHicassp,Tarokh,MaVe05} are
optimal. The complexity of algorithms from
\cite{FHicassp,Tarokh,MaVe05} is roughly $O(n^3)$. In a similar fashion one can, instead of using coding/decoding techniques associated with Reed/Solomon codes,
design the matrix and the corresponding recovery algorithm based on the techniques related to the coding/decoding of
Expander codes (see e.g.
\cite{XHexpander,JXHC08,InRu08} and references therein). In that case recovering $\x$ in
(\ref{eq:system}) is significantly faster for large dimensions $n$. Namely, the complexity of the techniques from e.g. \cite{XHexpander,JXHC08,InRu08}
(or their slight modifications) is usually
$O(n)$ which is clearly for large $n$ significantly smaller than $O(n^3)$. However,
the techniques based on coding/decoding of Expander codes usually do not allow for $\beta$ to be as large as
$\frac{\alpha}{2}$.

The main interest of this paper however will be the algorithms from the second class. Within the second class are the algorithms that should be designed without having the choice of $A$ (instead the matrix $A$ is rather given to us). Designing the algorithms from the second class is substantially harder compared to the design of the algorithms from the first class. The main reason for hardness is that when there is no choice in $A$ the recovery
problem (\ref{eq:system}) becomes NP-hard. The following two algorithms (and their different
variations) have been often viewed historically as solid heuristics for solving (\ref{eq:system}) (in recent years belief propagation type of algorithms are emerging as strong alternatives as well):
\begin{enumerate}
\item \underline{\emph{Orthogonal matching pursuit - OMP}}
\item \underline{\emph{Basis pursuit -
$\ell_1$-optimization.}}
\end{enumerate}
Under certain probabilistic assumptions on the elements of $A$ it can be shown (see e.g. \cite{JATGomp,JAT,NeVe07})
that if $m=O(k\log(n))$
OMP (or slightly modified OMP) can recover $\x$ in (\ref{eq:system})
with complexity of recovery $O(n^2)$. On the other hand a stage-wise
OMP from \cite{DTDSomp} recovers $\x$ in (\ref{eq:system}) with
complexity of recovery $O(n \log n)$. Somewhere in between OMP and BP are recent improvements CoSAMP (see e.g. \cite{NT08}) and Subspace pursuit (see e.g. \cite{DaiMil08}), which guarantee (assuming the linear regime) that the $k$-sparse $\x$ in (\ref{eq:system}) can be recovered in polynomial time with $m=O(k)$ equations.

We will now further narrow down our interest to only the performance of $\ell_1$-optimization. (Variations of the standard $\ell_1$-optimization from e.g.
\cite{CWBreweighted,SChretien08,SaZh08}) as well as those from \cite{SCY08,FL08,GN03,GN04,GN07,DG08} related to $\ell_q$-optimization, $0<q<1$
are possible as well.) Basic $\ell_1$-optimization algorithm finds $\x$ in
(\ref{eq:system}) by solving the following $\ell_1$-norm minimization problem
\begin{eqnarray}
\mbox{min} & & \|\x\|_{1}\nonumber \\
\mbox{subject to} & & A\x=\y. \label{eq:l1}
\end{eqnarray}
Due to its popularity the literature on the use of the above algorithm is rapidly growing. We below restrict our attention to two, in our mind, the most influential works that relate to (\ref{eq:l1}).

The first one is \cite{CRT} where the authors were able to show that if
$\alpha$ and $n$ are given, $A$ is given and satisfies the restricted isometry property (RIP) (more on this property the interested reader can find in e.g. \cite{Crip,CRT,Bar,Ver,ALPTJ09}), then
any unknown vector $\x$ with no more than $k=\beta n$ (where $\beta$
is a constant dependent on $\alpha$ and explicitly
calculated in \cite{CRT}) non-zero elements can be recovered by
solving (\ref{eq:l1}). As expected, this assumes that $\y$ was in
fact generated by that $\x$ and given to us. The case when the
available $\y$'s are noisy versions of real $\y$'s is also of
interest \cite{CT,CRT,HN,W}. Although that case is not of primary interest in the present paper
it is worth mentioning that the recent
popularity of $\ell_1$-optimization in compressed sensing is
significantly due to its robustness with respect to noisy
$\y$'s. (Of course, the main reason for its popularity is its
ability to solve (\ref{eq:system}) for a very wide range of matrices
$A$; more on this universality from a statistical point of view the interested reader can find in \cite{DonTan09Univ}.)

However, the RIP is only a \emph{sufficient}
condition for $\ell_1$-optimization to produce the $k$-sparse solution of
(\ref{eq:system}). Instead of characterizing $A$ through the RIP
condition, in \cite{DonohoUnsigned,DonohoPol} Donoho looked at its geometric properties/potential. Namely,
in \cite{DonohoUnsigned,DonohoPol} Donoho considered polytope obtained by
projecting the regular $n$-dimensional cross-polytope $C_p^n$ by $A$. He then established that
the solution of (\ref{eq:l1}) will be the $k$-sparse solution of
(\ref{eq:system}) if and only if
$AC_p^n$ is centrally $k$-neighborly
(for the definitions of neighborliness, details of Donoho's approach, and related results the interested reader can consult now already classic references \cite{DonohoUnsigned,DonohoPol,DonohoSigned,DT}). In a nutshell, using the results
of \cite{PMM,AS,BorockyHenk,Ruben,VS}, it is shown in
\cite{DonohoPol}, that if $A$ is a random $m\times n$
ortho-projector matrix then with overwhelming probability $AC_p^n$ is centrally $k$-neighborly (as usual, under overwhelming probability we in this paper assume
a probability that is no more than a number exponentially decaying in $n$ away from $1$). Miraculously, \cite{DonohoPol,DonohoUnsigned} provided a precise characterization of $m$ and $k$ (in a large dimensional context) for which this happens.

It should be noted that one usually considers success of
(\ref{eq:l1}) in recovering \emph{any} given $k$-sparse $\x$ in (\ref{eq:system}). It is also of interest to consider success of
(\ref{eq:l1}) in recovering
\emph{almost any} given $\x$ in (\ref{eq:system}). We below make a distinction between these
cases and recall on some of the definitions from
\cite{DonohoPol,DT,DTciss,DTjams2010,StojnicCSetam09,StojnicICASSP09}.

Clearly, for any given constant $\alpha\leq 1$ there is a maximum
allowable value of $\beta$ such that for \emph{any} given $k$-sparse $\x$ in (\ref{eq:system}) the solution of (\ref{eq:l1})
is with overwhelming probability exactly that given $k$-sparse $\x$. We will refer to this maximum allowable value of
$\beta$ as the \emph{strong threshold} (see
\cite{DonohoPol}). Similarly, for any given constant
$\alpha\leq 1$ and \emph{any} given $\x$ with a given fixed location of non-zero components and a given fixed combination of its elements signs
there will be a maximum allowable value of $\beta$ such that
(\ref{eq:l1}) finds that given $\x$ in (\ref{eq:system}) with overwhelming
probability. We will refer to this maximum allowable value of
$\beta$ as the \emph{weak threshold} and will denote it by $\beta_{w}$ (see, e.g. \cite{StojnicICASSP09,StojnicCSetam09}).

In our own work \cite{StojnicCSetam09} we provided a novel probabilistic framework for performance characterization of (\ref{eq:l1}) (the framework seems rather powerful; in fact, we found hardly any sparse type of problem that the framework was not able to handle with almost impeccable precision). Using that framework we obtained lower bounds on $\beta_w$. These lower bounds were in an excellent numerical agreement with the values obtained for $\beta_w$ in \cite{DonohoPol}. One would therefore be tempted to believe that our lower bounds from \cite{StojnicCSetam09} are tight. In this paper we design a mechanism that can be used to compute the upper bounds on $\beta_w$ (as it was the case with the framework of \cite{StojnicCSetam09}, the new framework does not seem to be restricted in any way to the $\ell_1$ type of sparsity). The obtained upper bounds will match the lower bounds computed in \cite{StojnicCSetam09} and essentially make them optimal. We should as an important side point mention that in a companion paper \cite{StojnicEquiv10}) we designed an approach that reveals even qualitative (not only numerical) agreement between the results from \cite{StojnicCSetam09} and those from \cite{DonohoPol}. Obviously, using the results of \cite{StojnicEquiv10} one can then also argue that the lower bounds of \cite{StojnicCSetam09} are optimal. The point of the present work, though, should be viewed in a much broader context. It gives a general framework for bounding thresholds without relying on making them equivalent to the known optimal ones. Or in other words, it extends the range of applicability to the cases where the optimal ones may not be known. A collection of interesting applications of this framework will be presented in a few forthcoming companion papers. 

We organize the rest of the paper in the following way. In Section
\ref{sec:theorems} we introduce two key theorems that will be the heart of our subsequent analysis.
In Section
\ref{sec:unsigned} we create the mechanism for computing the upper bounds on $\beta_w$ in the case of general sparse signals $\x$ for a class of random matrices $A$. In Section \ref{sec:signed} we will then specialize results from Section \ref{sec:unsigned} to the so-called signed vectors $\x$.
Finally, in Section \ref{sec:discuss} we discuss obtained results.

\section{Key theorems}
\label{sec:theorems}

In this section we introduce two useful theorems that will be of key importance in our subsequent
analysis. First we recall on a null-space characterization of
$A$ that guarantees that the solution of
(\ref{eq:l1}) is the $k$-sparse solution of (\ref{eq:system}). Moreover, the characterization will establish this for any $\beta n$-sparse $\x$ with a fixed location of nonzero components and a fixed combination of signs of its elements. Since the analysis will clearly be irrelevant with respect to what particular location and what particular combination of signs of nonzero elements are chosen, we can for the simplicity of the exposition and without loss of generality assume that the components $\x_{1},\x_{2},\dots,\x_{n-k}$ of $\x$ are equal to zero and the components $\x_{n-k+1},\x_{n-k+2},\dots,\x_n$ of $\x$ are smaller than or equal to zero. Moreover, throughout the paper we will call such an $\x$ $k$-sparse and non-positive. Under this assumption we have the following theorem from \cite{StojnicICASSP09} that provides
such a characterization (similar characterizations can be found in
\cite{DH01,FN,LN,Y,XHapp,SPH,DTbern}; furthermore, if instead of $\ell_1$ one, for
example, uses an $\ell_q$-optimization ($0<q<1$) in (\ref{eq:l1}) then
characterizations similar to the ones from
\cite{DH01,FN,LN,Y,XHapp,SPH,DTbern} can be derived as well
\cite{GN03,GN04,GN07}).
\begin{theorem}(Nonzero part of $\x$ has fixed signs and location)
Assume that an $m\times n$ matrix $A$ is given. Let $\x$
be a $k$-sparse non-positive vector. Also let $\x_1=\x_2=\dots=\x_{n-k}=0.$
Further, assume that $\y=A\x$ and that $\w$ is
an $n\times 1$ vector. If
\begin{equation}
(\forall \w\in \textbf{R}^n | A\w=0) \quad  \sum_{i=n-k+1}^n \w_i<\sum_{i=1}^{n-k}|\w_{i}|
\label{eq:thmeqgenweak1}
\end{equation}
then the solution of (\ref{eq:l1}) is $\x$. Moreover, if
\begin{equation}
(\exists \w\in \textbf{R}^n | A\w=0) \quad  \sum_{i=n-k+1}^n \w_i>\sum_{i=1}^{n-k}|\w_{i}|
\label{eq:thmeqgenweak2}
\end{equation}
then there will be a $k$-sparse nonnegative $\x$ that satisfies (\ref{eq:system}) and is not the solution of (\ref{eq:l1}).
\label{thm:thmgenweak}
\end{theorem}
\begin{proof}
The first part follows directly from Theorem $2$ in \cite{StojnicICASSP09}. For the completeness we just sketch the argument again. Let $\hatx$ be the solution of (\ref{eq:l1}). We want to show that if (\ref{eq:thmeqgenweak1}) holds then $\hatx=\x$. To that end assume opposite, i.e. assume that (\ref{eq:thmeqgenweak1}) holds but $\hatx\neq\x$. Then since $\y=A\hatx$ and $\y=A\x$ one must have $\hatx =\x+\w$ with $\w$ such that $A\w=0$. Also, since $\hatx$ is the solution of (\ref{eq:l1}) one has that
\begin{equation}
\sum_{i=1}^n|-\x_i-\w_i|\leq \sum_{i=1}^{n}|\x_i|.\label{eq:absval}
\end{equation}
Now, the key point observed for the first time in Theorem $2$ in \cite{StojnicICASSP09} comes into play. By just simply ``removing the absolute values" on the last $k$ elements of the sum on the left-hand side one has that the following must hold as well
\begin{equation}
\sum_{i=1}^{n-k} |-\x_i-\w_i|+\sum_{i=n-k+1}^n (-\x_i-\w_i)\leq \sum_{i=1}^{n}|\x_i|.\label{eq:noabsval}
\end{equation}
Since we assumed that $\x_i\leq 0,n-k+1\leq i\leq n$, and $\x_i=0,1\leq i\leq n-k$, one from (\ref{eq:noabsval}) obtains
\begin{equation}
\sum_{i=1}^{n-k} |\w_i|-\sum_{i=n-k+1}^n \w_i\leq 0.\label{eq:wcon}
\end{equation}
or equivalently
\begin{equation}
\sum_{i=1}^{n-k} |\w_i|\leq\sum_{i=n-k+1}^n \w_i.\label{eq:wcon1}
\end{equation}
Clearly, (\ref{eq:wcon1}) contradicts (\ref{eq:thmeqgenweak1}) and $\hatx\neq\x$ can not hold. Therefore $\hatx=\x$ which is exactly what the first part of the theorem claims.

For the ``moreover" part assume that (\ref{eq:thmeqgenweak2}) holds, i.e. we assume
\begin{equation}
(\exists \w\in \textbf{R}^n | A\w=0) \quad  \sum_{i=n-k+1}^n \w_i>\sum_{i=1}^{n-k}|\w_{i}|\label{eq:thmeqgenweak3}
\end{equation}
and want to show that there is a non-positive $k$-sparse $\x$ such that (\ref{eq:absval}) holds (with a strict inequality). This would imply that there is a non-positive $\x$ such that $A\x=\y$ and $\x$ is not the solution of (\ref{eq:l1}). Since (\ref{eq:wcon1}) is just rewritten (\ref{eq:thmeqgenweak3}) one can go backwards from (\ref{eq:wcon1}) to (\ref{eq:absval}) (just additionally making all the inequalities strict in the process). The only problem will happen in a backward jump from (\ref{eq:noabsval}) to (\ref{eq:absval}). That jump will be fine pretty much for any non-positive $k$-sparse $\x$ if all $\w_i$'s are negative. If some of them are not then one has to design a specific non-positive $k$-sparse $\x$ for which jump from (\ref{eq:noabsval}) to (\ref{eq:absval}) is justified or in other words for which (\ref{eq:noabsval}) and (\ref{eq:absval}) are equivalent. To that end let us assume that $\w_j>0$ for $j\in{\cal J},{\cal J}\subset \{n-k+1,n-k+2,\dots,n\}$ and let $\bar{\cal J}$ be such that ${\cal J}\bigcup \bar{\cal J}=\{n-k+1,n-k+2,\dots,n\}$ and ${\cal J}\bigcap \bar{\cal J}=\emptyset$. Then for $\x$ such that $\x_j=0$ for $1\leq j\leq n-k$,  $\x_j=-\w_j$ for $j\in{\cal J}$, and $\x_j\leq 0$ for $j\in \bar{\cal J}$ one has that (\ref{eq:thmeqgenweak3}) implies
\begin{equation}
\sum_{i=1}^n|-\x_i-\w_i|< \sum_{i=1}^{n}|\x_i|.\label{eq:absvalrev}
\end{equation}
or in other words that $\x$ can not be the solution of (\ref{eq:l1}). This concludes the proof of the second (``moreover") part.
\end{proof}

Before proceeding further we would like to say a few words about the above theorem. In our opinion the first part of the theorem that was put forth in \cite{StojnicICASSP09} is the unsung hero of all the success achieved in the thresholds analysis through various frameworks that we eventually designed. It was  fist recognized in \cite{StojnicICASSP09} that it could lead to the optimal performance characterizations of $\ell_1$-optimization.  However, the analysis in \cite{StojnicICASSP09} stopped somewhat short of the ultimate goal and it achieved only a moderate success in performance characterization of $\ell_1$-optimization (of course, not to take anything away from the strength of the first part of the above theorem, it was strong enough even back then when the framework of \cite{StojnicICASSP09} was designed; it is just that the author of \cite{StojnicICASSP09} was apparently not skilled enough to utilize it). Eventually, it was put to an ultimate utilization in \cite{StojnicCSetam09}. While the success of the framework designed in \cite{StojnicCSetam09} is a story on its own we feel that this simple ``removing absolute values" observation made in Theorem $2$ in \cite{StojnicICASSP09} is a remarkable piece of the mosaic that makes everything work to perfection.

Now, with regard to the second part of the above theorem, the story is somewhat similar. The first thing one should say is that the above proof of it is really nothing too original; it follows the well known ``converse" strategy of the corresponding proofs when the absolute values are present (see, e.g. \cite{DH01,GN03}). It is just that we never presented this ``remove the absolute values" observation in a converse way before. Basically, we did not find the second part of the theorem to be of any (let alone much) use if one were to create the lower bounds on the thresholds. However, as the reader might guess, if one is concerned with proving the upper bounds the second part of the above theorem becomes the same type of the unsung hero that the first one was for the success of the framework of \cite{StojnicCSetam09}. Below we use it to create a machinery as powerful as the one from \cite{StojnicCSetam09} that provides the corresponding framework for upper-bounding the thresholds.

Before moving to the design of the framework, we would also like to say a few words about a possible design of the matrix $A$ that would satisfy the conditions of Theorem \ref{thm:thmgenweak}. Designing matrix $A$ such that
(\ref{eq:thmeqgenweak1}) holds would not be that hard.
The problem is that one does not know \emph{a priori} which $k$ components of $\x$ will be nonzero and which signs they will have. That would essentially force one to design $A$ such that (\ref{eq:thmeqgenweak1}) holds for any subset of $\{1,2,\dots,n\}$ of cardinality $k$ and any combination of signs on that subset. If one
assumes that $m$ and $k$ are proportional to $n$ (the case of our
interest in this paper) this is an enormous combinatorial task and the construction of such a deterministic
matrix $A$ is clearly not easy (in fact, as observed in e.g. \cite{StojnicCSetam09} one may say that it is one of the most fundamental open problems in the area of theoretical compressed sensing; more on an equally important inverse problem of checking if a given matrix satisfies the condition of Theorem \ref{thm:thmgenweak} for any subset of $\{1,2,\dots,n\}$ of cardinality $k$ and any combination of signs, the interested reader can find in \cite{JudNem08,DaEl08}). On the other hand, turning to
random matrices significantly simplifies things. As we will see later in the paper, Gaussian random matrices $A$ will turn out to be a very convenient choice.
The following phenomenal result from \cite{Gordon88} that relates to such matrices will be the key ingredient in the analysis that will follow.
\begin{theorem}(\cite{Gordon88})
\label{thm:Gordonmesh1} Let $X_{ij}$ and $Y_{ij}$, $1\leq i\leq n,1\leq j\leq m$, be two centered Gaussian processes which satisfiy the following inequalities for all choices of indices
\begin{enumerate}
\item $E(X_{ij}^2)=E(Y_{ij}^2)$
\item $E(X_{ij}X_{ik})=E(Y_{ij}Y_{ik})$
\item $E(X_{ij}X_{lk})=E(Y_{ij}Y_{lk}), i\neq l$.
\end{enumerate}
Then
\begin{equation*}
P(\bigcap_{i}\bigcup_{j}(X_{ij}\geq \lambda_{ij}))\leq P(\bigcap_{i}\bigcup_{j}(Y_{ij}\geq \lambda_{ij})).
\end{equation*}
\end{theorem}

\section{Upper-bounding $\beta_w$ -- general $\x$} \label{sec:unsigned}

In this section we probabilistically analyze validity of the null-space characterization given in the second part of Theorem \ref{thm:thmgenweak}. Essentially, we will design a mechanism for computing upper bounds on $\beta_w$ (in fact, since it will be slightly more convenient we will actually determine lower bounds on $\alpha$; that is of course conceptually  the same as finding the upper-bounds on $\beta$).

We start by defining a quantity $\tau$ that will play one of the key roles below
\begin{eqnarray}
\tau(A) =  \min & & (\sum_{i=1}^{n-k}|\w_i|-\sum_{i=n-k+1}^{n}\w_i) \nonumber \\
\mbox{subject to} & &  A\w=0\nonumber \\
& & \|\w\|_2\leq 1.\label{eq:deftau}
\end{eqnarray}
Now, we will in the rest of the paper assume that the entries of $A$ i.i.d. standard normal random variables. Then one can say that for any $\alpha$ and $\beta$ for which
\begin{equation}
\lim_{n\rightarrow\infty}P(\tau(A)<0)=1, \label{eq:mcrit}
\end{equation}
there is a $k$-sparse $\x$ (from a set of $\x$'s with a given fixed location of nonzero components and a given fixed combination of their signs) which (\ref{eq:l1}) with probability $1$ fails to find. For a fixed $\beta$ our goal will be to find the largest possible $\alpha$ for which (\ref{eq:mcrit}) holds, i.e. for which (\ref{eq:l1}) fails with probability $1$.

Before going through the randomness of the problem and evaluation of $P(\tau(A)<0)$ we will try to provide a more explicit expression for $\tau$ then the one given by the optimization problem in (\ref{eq:deftau}). We proceed by slightly rephrasing (\ref{eq:deftau}):
\begin{eqnarray}
\tau(A) =  \min_{\t,\w} & & (\sum_{i=1}^{n-k}\t_i-\sum_{i=n-k+1}^{n}\w_i) \nonumber \\
\mbox{subject to} & & -\t_i\leq \w_i\leq \t_i, 1\leq i\leq n-k \nonumber \\
& &  A\w=0\nonumber \\
& & \|\w\|_2\leq 1.\label{eq:deftau1}
\end{eqnarray}
Now we write a partial dual over $\t$ of the optimization problem in (\ref{eq:deftau1}) (the strong duality trivially holds throughout the rest of this derivation)
\begin{eqnarray}
\tau(A)= \max_{\lambda^{(1)},\lambda^{(2)}}\min_{\t,\w} &  & \sum_{i=1}^{n-k}\t_i
-\sum_{i=n-k+1}^{n}\w_i+\sum_{i=1}^{n-k}\lambda_i^{(1)} (\w_i-\t_i)+\sum_{i=1}^{n-k}\lambda_i^{(2)} (-\w_i-\t_i)  \nonumber \\
\mbox{subject to} & & \lambda_i^{(1)}\geq 0, 1\leq i\leq n-k \nonumber \\
& & \lambda_i^{(2)}\geq 0, 1\leq i\leq n-k \nonumber \\
& &  A\w=0\nonumber \\
& & \|\w\|_2\leq 1.
\label{eq:deftau2}
\end{eqnarray}
After regrouping the terms one has
\begin{eqnarray}
\tau(A)= \max_{\lambda^{(1)},\lambda^{(2)}}\min_{\t,\w} &  & \sum_{i=1}^{n-k}\t_i(1-\lambda_i^{(1)}-\lambda_i^{(2)})
-\sum_{i=n-k+1}^{n}\w_i+\sum_{i=1}^{n-k}\lambda_i^{(1)} \w_i-\sum_{i=1}^{n-k}\lambda_i^{(2)}\w_i \nonumber \\
\mbox{subject to} & & \lambda_i^{(1)}\geq 0, 1\leq i\leq n-k \nonumber \\
& & \lambda_i^{(2)}\geq 0, 1\leq i\leq n-k\nonumber \\
& &  A\w=0\nonumber \\
& & \|\w\|_2\leq 1.
\label{eq:deftau3}
\end{eqnarray}
To make the inner minimization over $\t$ bounded one must have $(1-\lambda_i^{(1)}-\lambda_i^{(2)})=0$ for any $1\leq i\leq (n-k)$. Replacing that back in (\ref{eq:deftau3}) we obtain
\begin{eqnarray}
\tau(A)= \max_{\lambda^{(1)},\lambda^{(2)}}\min_{\w} &  &
-\sum_{i=n-k+1}^{n}\w_i+\sum_{i=1}^{n-k}\lambda_i^{(1)} \w_i-\sum_{i=1}^{n-k}\lambda_i^{(2)}\w_i  \nonumber \\
\mbox{subject to} & & \lambda_i^{(1)}\geq 0, 1\leq i\leq n-k \nonumber \\
& & \lambda_i^{(2)}\geq 0, 1\leq i\leq n-k\nonumber \\
& & 1-\lambda_i^{(1)}-\lambda_i^{(2)}=0, 1\leq i\leq n-k\nonumber \\
& &  A\w=0\nonumber \\
& & \|\w\|_2\leq 1.
\label{eq:deftau4}
\end{eqnarray}
We now write the remaining part of the dual over $\w$ (this part of the dual could have been written together with the first one over $\t$; to make the expressions lighter we split the writing in two steps)
\begin{eqnarray}
\tau(A)= \max_{\lambda^{(1)},\lambda^{(2)},\nu,\gamma}\min_{\w} &  &
-\sum_{i=n-k+1}^{n}\w_i+\sum_{i=1}^{n-k}\lambda_i^{(1)} \w_i-\sum_{i=1}^{n-k}\lambda_i^{(2)}\w_i+\nu^T A\w +\gamma\sum_{i=1}^{n}\w_i^2-\gamma  \nonumber \\
\mbox{subject to} & & \lambda_i^{(1)}\geq 0, 1\leq i\leq n-k \nonumber \\
& & \lambda_i^{(2)}\geq 0, 1\leq i\leq n-k\nonumber \\
& & 1-\lambda_i^{(1)}-\lambda_i^{(2)}=0, 1\leq i\leq n-k.
\label{eq:deftau5}
\end{eqnarray}
Further, using the last constraint one can get rid of one of the $\lambda$'s (say $\lambda^{(1)}$) and obtain
\begin{eqnarray}
\tau(A)= \max_{\lambda^{(2)},\nu,\gamma}\min_{\w} &  &
-\sum_{i=n-k+1}^{n}\w_i+\sum_{i=1}^{n-k}(1-2\lambda_i^{(2)})\w_i+\nu^T A\w +\gamma\sum_{i=1}^{n}\w_i^2-\gamma  \nonumber \\
\mbox{subject to} & & \lambda_i^{(2)}\leq 1, 1\leq i\leq n-k \nonumber \\
& & \lambda_i^{(2)}\geq 0, 1\leq i\leq n-k.\nonumber \\
\label{eq:deftau6}
\end{eqnarray}
At this point we proceed by solving the inner minimization over $\w$. To that end, let
\begin{equation}
f_1(\lambda^{(2)},\nu,\gamma,\w)=-\sum_{i=n-k+1}^{n}\w_i+\sum_{i=1}^{n-k}(1-2\lambda_i^{(2)})\w_i+\nu^T A\w +\gamma\sum_{i=1}^{n}\w_i^2-\gamma.\label{eq:deftau7}
\end{equation}
Since $f_1(\cdot)$ is convex in $\w$ we simply find the optimal $\w$ by equaling the derivative of $f_1(\cdot)$ with respect to $\w$ to zero. We then have
\begin{eqnarray}
\frac{df_1(\lambda^{(2)},\nu,\gamma,\w)}{d\w_i} & = & (1-2\lambda_i^{(2)})+\nu^T A_i +2\gamma\w_i, 1\leq i\leq n-k\nonumber \\
\frac{df_1(\lambda^{(2)},\nu,\gamma,\w)}{d\w_i} & = & -1+\nu^T A_i +2\gamma\w_i, n-k+1\leq i\leq n
\label{eq:deftau8}
\end{eqnarray}
where as expected $A_i$ is the $i$-th column of $A$. Let
\begin{equation}
\z=[(1-2\lambda_i^{(2)}),(1-2\lambda_2^{(2)}),\dots,(1-2\lambda_{n-k}^{(2)}),-1,-1,\dots,-1]^T.\label{eq:defz}
\end{equation}
From (\ref{eq:deftau8}) one easily finds
\begin{eqnarray}
\w_{opt}=\frac{-\z-A^T\nu}{2\gamma}.\label{eq:solw}
\end{eqnarray}
Removing the inner minimization over $\w$ in (\ref{eq:deftau6}) and recognizing the relation between $\z$ and $\lambda^{(2)}$ we have
\begin{eqnarray}
\tau(A)= \max_{\z,\nu,\gamma}&  &
(\z-A^T \nu)^T\w_{opt} +\gamma\|\w_{opt}\|_2^2-\gamma  \nonumber \\
\mbox{subject to}& & |\z_i|\leq 1, 1\leq i\leq n-k.\nonumber \\
& & \z_i=-1, n-k+1\leq i\leq n.
\label{eq:deftau9}
\end{eqnarray}
Finally after plugging $\w_{opt}$ from (\ref{eq:solw}) in (\ref{eq:deftau9}) we obtain
\begin{eqnarray}
\tau(A)= \max_{\z,\nu,\gamma}&  &
-\frac{\|\z-A^T \nu\|_2^2}{4\gamma}-\gamma  \nonumber \\
\mbox{subject to}& & |\z_i|\leq 1, 1\leq i\leq n-k.\nonumber \\
& & \z_i=-1, n-k+1\leq i\leq n.
\label{eq:deftau10}
\end{eqnarray}
The maximization over $\gamma$  is then trivial and one finally has
\begin{eqnarray}
\tau(A)= \max_{\z,\nu}&  &
-\|\z-A^T \nu\|_2  \nonumber \\
\mbox{subject to}& & |\z_i|\leq 1, 1\leq i\leq n-k\nonumber \\
& & \z_i=-1, n-k+1\leq i\leq n.
\label{eq:deftau11}
\end{eqnarray}
or in a more convenient form
\begin{eqnarray}
\tau(A)= -\min_{\z,\nu}&  &
\|\z-A^T \nu\|_2  \nonumber \\
\mbox{subject to}& & |\z_i|\leq 1, 1\leq i\leq n-k\nonumber \\
& & \z_i=-1, n-k+1\leq i\leq n.
\label{eq:deftau12}
\end{eqnarray}
At this point we are almost ready to switch to the probabilistic aspect of the analysis. To that end we do the last piece of transformation. Namely, we rewrite (\ref{eq:deftau12}) as
\begin{eqnarray}
\tau(A)= -\min_{\z,\nu} \max_{\|\a\|_2=1} &  &
\a^T(\z-A^T \nu)  \nonumber \\
\mbox{subject to}& & |\z_i|\leq 1, 1\leq i \leq n-k\nonumber \\
& & \z_i=-1, n-k+1\leq i\leq n.
\label{eq:deftau13}
\end{eqnarray}
Now we are ready to invoke the results from Theorem \ref{thm:Gordonmesh1}. We do so through the following lemma which is slightly modified Lemma 3.1 from \cite{Gordon88} (Lemma 3.1 is a direct consequence of Theorem \ref{thm:Gordonmesh1} and the backbone of the escape through a mesh theorem utilized in \cite{StojnicCSetam09}).
\begin{lemma}
Let $A$ be an $m\times n$ matrix with i.i.d. standard normal components. Let $\g$ and $\h$ be $n\times 1$ and $m\times 1$ vectors, respectively, with i.i.d. standard normal components. Also, let $g$ be a standard normal random variable and let $Z$ be a set such that $Z=(\z|\z_i=-1,n-k+1\leq i\leq n \quad \mbox{and} \quad |\z_i|\leq 1, 1\leq i\leq n-k)$. Then
\begin{equation}
P(\min_{\z\in Z,\nu\in R^n\setminus 0}\max_{\|\a\|_2=1}(-\a^T A^T\nu +\|\nu\|_2 g-\zeta_{\a,\z,\nu})\geq 0)\geq P(\min_{\z\in Z,\nu\in R^n\setminus 0}\max_{\|\a\|_2=1}(\|\nu\|_2\sum_{i=1}^{m}\g_i\a_i+\sum_{i=1}^{n}\h_i\nu_i-\zeta_{\a,\z,\nu})\geq 0).\label{eq:problemma}
\end{equation}\label{eq:unsignedlemma}
\end{lemma}
\begin{proof}
The proof is exactly the same as the one of Lemma 3.1 in \cite{Gordon88}. The only difference is that one should make identical copies over $\z$ of the processes $X_{x,y},Y_{x,y}$ defined in that proof. The rest of the proof remains unaltered.
\end{proof}
Let $\zeta_{\a,\z,\nu}=\epsilon_{5}^{(g)}\sqrt{n}\|\nu\|_2-\a^T\z$ with $\epsilon_{5}^{(g)}>0$ being an arbitrarily small constant independent of $n$. Then the left-hand side of the inequality in (\ref{eq:problemma}) is then the following probability of interest
\begin{equation*}
P(\min_{\z\in Z,\nu\in R^n\setminus 0}\max_{\|\a\|_2=1}(\|\nu\|_2\sum_{i=1}^{n}\g_i\a_i+\sum_{i=1}^{m}\h_i\nu_i-\epsilon_{5}^{(g)}\sqrt{n}\|\nu\|_2+\a^T\z)\geq 0).
\end{equation*}
After solving the inner maximization over $\a$ and pulling out $\|\nu\|_2$ one has
\begin{equation*}
P(\min_{\z\in Z,\nu\in R^n\setminus 0}(\|\g+\frac{1}{\|\nu\|_2}\z\|_2+\sum_{i=1}^{m}\h_i\frac{\nu_i}{\|\nu\|_2}-\epsilon_{5}^{(g)}\sqrt{n})\geq 0).
\end{equation*}
Minimization of the second term then gives us
\begin{equation}
P(\min_{\z\in Z,\nu\in R^n\setminus 0}(\|\g+\frac{1}{\|\nu\|_2}\z\|_2)\geq \|\h\|_2+\epsilon_{5}^{(g)}\sqrt{n}).\label{eq:probanal1}
\end{equation}
Since $\h$ is a vector of $m$ i.i.d. standard normal variables it is rather trivial that $P(\|\h\|_2<(1+\epsilon_{1}^{(m)})\sqrt{m})\geq 1-e^{-\epsilon_{2}^{(m)} m}$ where $\epsilon_{1}^{(m)}>0$ is an arbitrarily small constant and $\epsilon_{2}^{(m)}$ is a constant dependent on $\epsilon_{1}^{(m)}$ but independent of $n$. Then from (\ref{eq:probanal1}) one obtains
\begin{multline}
P(\min_{\z\in Z,\nu\in R^n\setminus 0}(\|\g+\frac{1}{\|\nu\|_2}\z\|_2)\geq \|\h\|_2+\epsilon_{5}^{(g)}\sqrt{n})\\\geq (1-e^{-\epsilon_{2}^{(m)} m})
P(\min_{\z\in Z,\nu\in R^n\setminus 0}(\|\g+\frac{1}{\|\nu\|_2}\z\|_2)\geq (1+\epsilon_{1}^{(m)})\sqrt{m}+\epsilon_{5}^{(g)}\sqrt{n})).\label{eq:probanal2}
\end{multline}
Now, let $\bar{\g}=[\g_{(1)},\g_{(2)},\dots,\g_{(n-k)},\g_{n-k+1},\g_{n-k+2},\dots,\g_n]^T$, where $[\g_{(1)},\g_{(2)},\dots,\g_{(n-k)}]$ are magnitudes of $[\g_{1},\g_{2},\dots,\g_{n-k}]$ sorted in increasing order. Then clearly,
\begin{equation}
\min_{\z\in Z,\nu\in R^n\setminus 0}(\|\g+\frac{1}{\|\nu\|_2}\z\|_2)=\min_{\z\in |Z|,\nu\in R^n\setminus 0}(\|\bar{\g}-\frac{1}{\|\nu\|_2}\z\|_2)\label{eq:probanal3}
\end{equation}
where $|Z|=(\z|\z_i=-1,n-k+1\leq i\leq n \quad \mbox{and} \quad 0\leq \z_i\leq 1, 1\leq i\leq n-k)$. Moreover, the optimization on the right-hand side of (\ref{eq:probanal2}) is structurally the same as the one in equation $(15)$ in \cite{StojnicCSetam09} (actually to be more precise it is the same as the weak threshold equivalent to $(15)$). Essentially, the exact equivalence between these optimizations is achieved after in $(15)$ from \cite{StojnicCSetam09}   $\tilde{\h}$ is replaced by $\bar{\g}$, $\nu$ is replaced by $\frac{1}{\|\nu\|_2}$, $\lambda$ is restricted to the lower $(n-k)$ components, and after one additionally notes that in $(15)$ from \cite{StojnicCSetam09} $0\leq\lambda\leq \nu$, which corresponds to $0\leq \z_i\leq 1, 1\leq i\leq n-k$ introduced above (that way one would in essence obtain the weak threshold equivalent to $(15)$; this was not explicitly written anywhere in \cite{StojnicCSetam09} but is rather obvious; in \cite{StojnicCSetam09} we, instead, made a ``weak" equivalence to its $(29)$). With these replacements one can then use the machinery of \cite{StojnicCSetam09} to establish
\begin{equation}
\min_{\z\in |Z|,\nu\in R^n\setminus 0}(\|\bar{\g}+\frac{1}{\|\nu\|_2}\z\|_2)=\sqrt{\sum_{i=\cweak+1}^n\bar{\g}_i^2-\frac{((\bar{\g}
^T\z)-\sum_{i=1}^{\cweak}\bar{\g}_i)^2}{n-\cweak
}}=f_{\g}(\cweak)\label{eq:probanal4}
\end{equation}
where $\cweak$ is the solution of
\begin{equation}
\frac{(\bar{\g}
^T\z)-\sum_{i=1}^{\cweak}\bar{\g}_i}{n-\cweak
}=\bar{\g}_{\cweak}.\label{eq:defcw}
\end{equation}
As a side remark, we should point out that the key point to the success of our method is that the derivation of \cite{StojnicCSetam09}
establishes the equality in (\ref{eq:probanal4}). It is just that in \cite{StojnicCSetam09} only the ``smaller than" inequality part of this equality was utilized. At this point we have established the core of our upper-bounding arguments. The rest is just a slightly modified repetition of the derivations from \cite{StojnicCSetam09} so that we can make everything precise.

First we will define two quantities $c_{w}^{(l)}$ and $c_{w}^{(u)}$ as the solutions of the following two equations:
\begin{eqnarray}
\frac{(1-\epsilon_1^{(c)})E((\bar{\g}^T\z)-\sum_{i=1}^{\cweak^{(l)}} \bar{\g}_i)}{n-\cweak^{(l)}}-F_a^{-1}\left (\frac{(1+\epsilon_1^{(c)}\cweak^{(l)}}{n(1-\beta_w)}\right ) & = & 0\nonumber \\
\frac{(1+\epsilon_2^{(c)})E((\bar{\g}^T\z)-\sum_{i=1}^{\cweak^{(u)}} \bar{\g}_i)}{n-\cweak^{(u)}}-F_a^{-1}\left (\frac{(1-\epsilon_2^{(c)}\cweak^{(u)}}{n(1-\beta_w)}\right ) & = & 0.\label{eq:probanal44}
\end{eqnarray}
where $F_a^{-1}(\cdot)$ is the inverse cdf of the random variable $|X|$, $X$ is the standard normal random variable, and $\epsilon_{i}^{(c)}>0,1\leq i\leq 2$ are arbitrarily small constants independent of $n$. It follows then directly from the derivation $(32)-(39)$ in \cite{StojnicCSetam09} that
\begin{equation}
P(c_w\in\{c_w^{(l)},c_w^{(u)}\})\geq 1-e^{-\epsilon_{3}^{(c)}n}\label{eq:probcw}
\end{equation}
where $\epsilon_{3}^{(c)}$ is a constant dependent on $\epsilon_{i}^{(c)}>0,1\leq i\leq 2$, $\cweak^{(l)}$, $\cweak^{(u)}$ but independent of $n$.
We now set $c_w=c_w^{(u)}$ and focus on (\ref{eq:probanal4}). Concentration analysis machinery of \cite{StojnicCSetam09} will help us establish a ``high probability" lower bound on $f_{\g}(c_w)$ (this will amount to nothing but reversing the concentration arguments that we have established in \cite{StojnicCSetam09}; concentration arguments are of course easy to reverse; what was harder to reverse was the part before (\ref{eq:probanal4})). We now split $f_{\g}(c_w)$ into two parts i.e.
\begin{equation}
f_{\g}(c_w)=f_{\g}^{(1)}(\cweak)-f_{\g}^{(2)}(\cweak),\label{eq:deff}
\end{equation}
where $f_{\g}^{(1)}(\cweak)=\sum_{i=\cweak+1}^n\bar{\g}_i^2$ and $f_{\g}^{(2)}(\cweak)=(\bar{\g}
^T\z)-\sum_{i=1}^{\cweak}\bar{\g}_i)$. Now, $f_{\g}^{(1)}(\cweak)$ concentrates trivially, the argument is the same as the one that can be established when $\cweak=0$ (alternatively one can repeat derivation $(37)$ from \cite{StojnicCSetam09} to obtain the Lipschitz constant and combine it with Lipschitz concentration formula $(35)$ also in \cite{StojnicCSetam09}). So we have
\begin{equation}
P(f_{\g}^{(1)}(\cweak)\geq (1-\epsilon_{1}^{(g)})Ef_{\g}^{(1)}(\cweak))>1-e^{-\epsilon_{2}^{(g)}n},\label{eq:concf1}
\end{equation}
again as usual $\epsilon_{1}^{(g)}>0$ is an arbitrarily small constant and $\epsilon_{2}^{(g)}$ is a constant dependent on $\epsilon_{1}^{(g)}$ and $\cweak$ but independent of $n$. On the other hand, concentration of $f_{\g}^{(2)}(\cweak)$ follows by reversing the $(38)$ from \cite{StojnicCSetam09}, i.e.
\begin{equation}
P(f_{\g}^{(2)}(\cweak)\geq (1+\epsilon_{3}^{(g)})Ef_{\g}^{(2)}(\cweak))>1-e^{-\epsilon_{4}^{(g)}n}\label{eq:concf2}
\end{equation}
where again as usual $\epsilon_{3}^{(g)}>0$ is an arbitrarily small constant and $\epsilon_{4}^{(g)}$ is a constant dependent on $\epsilon_{3}^{(g)}$ and $\cweak$ but independent of $n$. Combination of (\ref{eq:probanal4}), (\ref{eq:concf1}) and (\ref{eq:concf1}) gives (the only other thing one should observe here is that $E((\bar{\g}
^T\z)-\sum_{i=1}^{\cweak}\bar{\g}_i)\geq 0$)
\begin{multline}
P\left (\sqrt{\sum_{i=\cweak+1}^n\bar{\g}_i^2-\frac{((\bar{\g}
^T\z)-\sum_{i=1}^{\cweak}\bar{\g}_i)^2}{n-\cweak
}}\geq \sqrt{(1-\epsilon_{1}^{(g)})E\sum_{i=\cweak+1}^n\bar{\g}_i^2-\frac{(1+\epsilon_{3}^{(g)})^2(E((\bar{\g}
^T\z)-\sum_{i=1}^{\cweak}\bar{\g}_i))^2}{n-\cweak
}}\right )\\\geq (1-e^{-\epsilon_{2}^{(g)}n})(1-e^{-\epsilon_{4}^{(g)}n}).\label{eq:probanal5}
\end{multline}
%
Now, let
\begin{equation}
m_{w}=\frac{1}{(1+\epsilon_{1}^{(m)})^2}(\sqrt{(1-\epsilon_{1}^{(g)})E\sum_{i=\cweak+1}^n\bar{\g}_i^2-\frac{(1+\epsilon_{3}^{(g)})^2(E((\bar{\g}
^T\z)-\sum_{i=1}^{\cweak}\bar{\g}_i))^2}{n-\cweak
}}-\epsilon_{3}^{(m)}\sqrt{n}-\epsilon\sqrt{n})^2,\label{eq:defmmax}
\end{equation}
where $\epsilon_{3}^{(m)}>0$ is an arbitrarily small constant. Combining (\ref{eq:probanal2}), (\ref{eq:probanal4}), (\ref{eq:probcw}), (\ref{eq:probanal5}), and  (\ref{eq:defmmax}) we have
\begin{equation}
P(\min_{\z\in Z,\nu\in R^n\setminus 0}(\|\g+\frac{1}{\|\nu\|_2}\z\|_2)\geq \|\h\|_2+\epsilon_{5}^{(g)}\sqrt{n})\\ \geq (1-e^{-\epsilon_{2}^{(m)}m})
(1-e^{-\epsilon_{2}^{(g)} n})(1-e^{-\epsilon_{4}^{(g)} n})(1-e^{-\epsilon_{3}^{(c)}n}).\label{eq:probanal6}
\end{equation}
Further combination of (\ref{eq:problemma}), (\ref{eq:probanal1}), (\ref{eq:probanal2}), and (\ref{eq:probanal6}) gives us that if $m=m_{w}$
\begin{equation}
\hspace{-.3in}P(\min_{\z\in Z,\nu\in R^n\setminus 0}\max_{\|\a\|_2=1}(-\a^T A\nu+\a^T\z+\|\nu\|_2(g-\epsilon_{5}^{(g)}\sqrt{n}))\geq 0)\geq (1-e^{-\epsilon_{2}^{(m)}m_w})
(1-e^{-\epsilon_{2}^{(g)} n})(1-e^{-\epsilon_{4}^{(g)} n})(1-e^{-\epsilon_{3}^{(c)}n}).\label{eq:probanal7}
\end{equation}
Since $P(g\leq\epsilon_{5}^{(g)}\sqrt{n})\geq 1-e^{-\epsilon_{6}^{(g)} n}$ (where $\epsilon_{6}^{(g)}$ is, as all other $\epsilon$'s in this paper are, independent of $n$) from (\ref{eq:probanal7}) we finally have
\begin{equation}
P(\min_{\z\in Z,\nu\in R^n\setminus 0}\max_{\|\a\|_2=1}(-\a^T A\nu+\a^T\z)> 0)\geq (1-e^{-\epsilon_{2}^{(m)}m_w})
(1-e^{-\epsilon_{2}^{(g)} n})(1-e^{-\epsilon_{4}^{(g)} n})(1-e^{-\epsilon_{6}^{(g)} n})(1-e^{-\epsilon_{3}^{(c)}n}).\label{eq:probanal8}
\end{equation}
Connecting (\ref{eq:deftau13}) and (\ref{eq:probanal8}) we obtain
\begin{equation*}
P(-\tau(A)> 0)\geq (1-e^{-\epsilon_{2}^{(m)}m_w})
(1-e^{-\epsilon_{2}^{(g)} n})(1-e^{-\epsilon_{4}^{(g)} n})(1-e^{-\epsilon_{6}^{(g)} n})(1-e^{-\epsilon_{3}^{(c)}n}),
\end{equation*}
and ultimately
\begin{equation}
\lim_{n\rightarrow\infty}P(\tau(A)< 0)=\lim_{n\rightarrow\infty} (1-e^{-\epsilon_{2}^{(m)}m_w})
(1-e^{-\epsilon_{2}^{(g)} n})(1-e^{-\epsilon_{4}^{(g)} n})(1-e^{-\epsilon_{6}^{(g)} n})(1-e^{-\epsilon_{3}^{(c)}n})=1\label{eq:finaltau}
\end{equation}
which is what we established as a goal in (\ref{eq:mcrit}).
We summarize the results in the following theorem.

\begin{theorem}(Exact weak threshold)
Let $A$ be an $m\times n$ matrix in (\ref{eq:system})
with i.i.d. standard normal components. Let
the unknown $\x$ in (\ref{eq:system}) be $k$-sparse. Further, let the location and signs of nonzero elements of $\x$ be arbitrarily chosen but fixed.
Let $k,m,n$ be large
and let $\alpha=\frac{m}{n}$ and $\beta_w=\frac{k}{n}$ be constants
independent of $m$ and $n$. Let $\erfinv$ be the inverse of the standard error function associated with zero-mean unit variance Gaussian random variable.  Further,
let all $\epsilon$'s below be arbitrarily small constants.
\begin{enumerate}
\item Let $\htheta_w$, ($\beta_w\leq \htheta_w\leq 1$) be the solution of
\begin{equation}
(1-\epsilon_{1}^{(c)})(1-\beta_w)\frac{\sqrt{\frac{2}{\pi}}e^{-(\erfinv(\frac{1-\theta_w}{1-\beta_w}))^2}}{\theta_w}-\sqrt{2}\erfinv ((1+\epsilon_{1}^{(c)})\frac{1-\theta_w}{1-\beta_w})=0.\label{eq:thmweaktheta}
\end{equation}
If $\alpha$ and $\beta_w$ further satisfy
\begin{equation}
\alpha>\frac{1-\beta_w}{\sqrt{2\pi}}\left (\sqrt{2\pi}+2\frac{\sqrt{2(\erfinv(\frac{1-\htheta_w}{1-\beta_w}))^2}}{e^{(\erfinv(\frac{1-\htheta_w}{1-\beta_w}))^2}}-\sqrt{2\pi}
\frac{1-\htheta_w}{1-\beta_w}\right )+\beta_w
-\frac{\left ((1-\beta_w)\sqrt{\frac{2}{\pi}}e^{-(\erfinv(\frac{1-\hat{\theta}_w}{1-\beta_w}))^2}\right )^2}{\hat{\theta}_w}\label{eq:thmweakalpha}
\end{equation}
then with overwhelming probability the solution of (\ref{eq:l1}) is the $k$-sparse $\x$ from (\ref{eq:system}).
\item Let $\htheta_w$, ($\beta_w\leq \htheta_w\leq 1$) be the solution of
\begin{equation}
(1+\epsilon_{2}^{(c)})(1-\beta_w)\frac{\sqrt{\frac{2}{\pi}}e^{-(\erfinv(\frac{1-\theta_w}{1-\beta_w}))^2}}{\theta_w}-\sqrt{2}\erfinv ((1-\epsilon_{2}^{(c)})\frac{1-\theta_w}{1-\beta_w})=0.\label{eq:thmweaktheta1}
\end{equation}
If on the other hand $\alpha$ and $\beta_w$ satisfy
\begin{multline}
\hspace{-.5in}\alpha<\frac{1}{(1+\epsilon_{1}^{(m)})^2}\left ((1-\epsilon_{1}^{(g)})(\htheta_w+\frac{2(1-\beta_w)}{\sqrt{2\pi}} \frac{\sqrt{2(\erfinv(\frac{1-\htheta_w}{1-\beta_w}))^2}}{e^{(\erfinv(\frac{1-\htheta_w}{1-\beta_w}))^2}})
-\frac{\left ((1-\beta_w)\sqrt{\frac{2}{\pi}}e^{-(\erfinv(\frac{1-\hat{\theta}_w}{1-\beta_w}))^2}\right )^2}{\hat{\theta}_w(1+\epsilon_{3}^{(g)})^{-2}}\right )\label{eq:thmweakalpha}
\end{multline}
then with overwhelming probability there will be a $k$-sparse $\x$ (from a set of $\x$'s with fixed locations and signs of nonzero components) that satisfies (\ref{eq:system}) and is \textbf{not} the solution of (\ref{eq:l1}).
\end{enumerate}
\label{thm:thmweakthr}
\end{theorem}
\begin{proof}
The first part was established in \cite{StojnicCSetam09}. The second part follows from the previous discussion combining (\ref{eq:thmeqgenweak2}), (\ref{eq:deftau}), (\ref{eq:mcrit}), (\ref{eq:probanal44}), (\ref{eq:defmmax}),  and (\ref{eq:finaltau}).
\end{proof}

\noindent While the previous theorem insists on precision one can do what we will refer to as the ``deepsilonification" and obtain a way more convenient characterization. After removing all $\epsilon$'s (or say after setting them to the values that are so small when compared to $n$ that on any available finite precision machine they don't impact the above characterization) one in a more informal language then has.

\noindent Assume the setup of the above theorem. Let $\alpha_w$ and $\beta_w$ satisfy the following:

\noindent \underline{\underline{\textbf{Fundamental characterization of the $\ell_1$ performance:}}}

\begin{center}
\shadowbox{$
(1-\beta_w)\frac{\sqrt{\frac{2}{\pi}}e^{-(\erfinv(\frac{1-\alpha_w}{1-\beta_w}))^2}}{\alpha_w}-\sqrt{2}\erfinv (\frac{1-\alpha_w}{1-\beta_w})=0.
$}
-\vspace{-.5in}\begin{equation}
\label{eq:thmweaktheta2}
\end{equation}
\end{center}

Then:
\begin{enumerate}
\item If $\alpha>\alpha_w$ then with overwhelming probability the solution of (\ref{eq:l1}) is the $k$-sparse $\x$ from (\ref{eq:system}).
\item If $\alpha<\alpha_w$ then with overwhelming probability there will be a $k$-sparse $\x$ (from a set of $\x$'s with fixed locations and signs of nonzero components) that satisfies (\ref{eq:system}) and is \textbf{not} the solution of (\ref{eq:l1}).
    \end{enumerate}

As stated above equation (\ref{eq:thmweaktheta2}) is the fundamental characterization of the $\ell_1$ performance. Numerical values of the weak threshold obtained using (\ref{eq:thmweaktheta2}) were presented in \cite{StojnicCSetam09}. As it was demonstrated there, the lower bounds on the thresholds were in an excellent numerical agreement with the optimal thresholds computed in \cite{DonohoPol,DonohoUnsigned}. Theorem  \ref{thm:thmweakthr} establishes that the lower bounds computed in \cite{StojnicCSetam09} (essentially those one can compute from (\ref{eq:thmweaktheta2})) are actually the upper bounds as well and as such are the exact values of the weak thresholds. Moreover, in a companion paper \cite{StojnicEquiv10} we established a qualitative equivalence of the characterization given in (\ref{eq:thmweaktheta2}) and the results obtained in \cite{DonohoPol}. It is rather fascinating to us how well the axiomatic system of mathematics works and that two so seemingly different approaches, the geometric one from \cite{DonohoPol} and the purely probabilistic one from \cite{StojnicCSetam09}, result in exactly the same optimal characterization of the performance of $\ell_1$-optimization.

Out of respect for an incredible effort that was put forth to characterize the $\ell_1$ performance in  \cite{DonohoUnsigned,DonohoPol,StojnicCSetam09,StojnicEquiv10} we present in Figure \ref{fig:weakthr} again the plot obtained based on the ultimate characterization (\ref{eq:thmweaktheta2}).
\begin{figure}[htb]
\centering
\centerline{\epsfig{figure=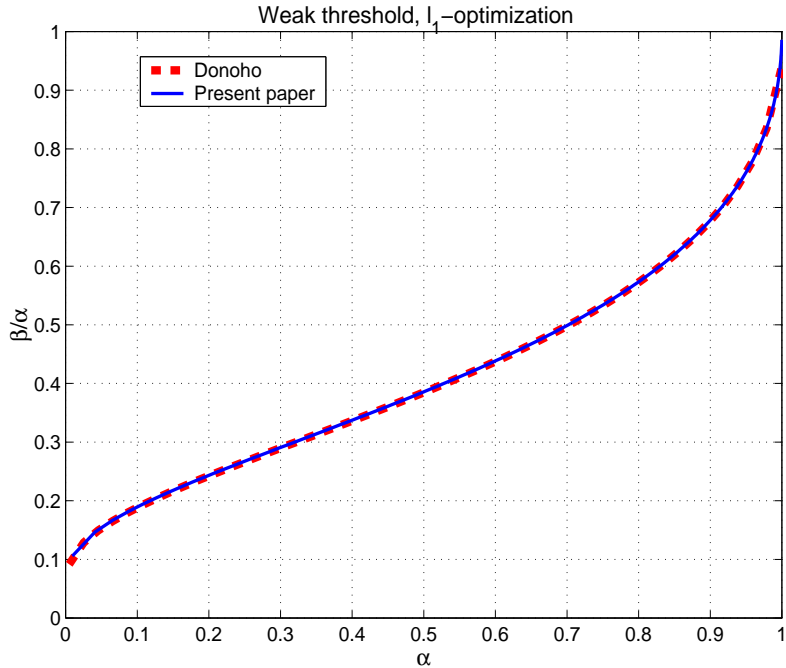,width=10.5cm,height=9cm}}
\vspace{-0.2in} \caption{\emph{Weak} threshold, $\ell_1$-optimization --- ultimate performance}
\label{fig:weakthr}
\end{figure}

\section{Upper-bounding $\beta_w$ -- signed $\x$} \label{sec:signed}

In this section we specialize the results from the previous section to the recovery of vectors $\x$ with elements known to have certain sign pattern.
Without loss of generality we assume that it is known that $\x_i\geq 0,1\leq i\leq n$. We also again assume that $\x$ is $k$-sparse, i.e. we assume that $\x$ has no more than $k$ nonzero elements. To solve (\ref{eq:system}) for such an $\x$ instead of (\ref{eq:l1}) we consider the following optimization problem (see, e.g. \cite{DonohoUnsigned,DT,StojnicCSetam09})
\begin{eqnarray}
\mbox{min} & & \|\x\|_{1}\nonumber \\
\mbox{subject to} & & A\x=\y\nonumber \\
& & \x_i\geq 0. \label{eq:l1signed}
\end{eqnarray}
In what follows we will determine the upper bound on the weak threshold that characterizes the performance of the above algorithm. Before proceeding further we quickly recall on and readjust the definition of the weak threshold. The definition of the weak threshold was already introduced in Section \ref{sec:back} when recovery of general signals (vectors) $\x$ was considered. Here, we slightly modify it so that it fits the scenario of a priori known sign patterns of elements of $\x$. Namely, for a given $\alpha$, $\betawnon$ is the maximum value of $\beta$ such that the solution (\ref{eq:l1signed}) is the $\beta n$-sparse solution of (\ref{eq:system}) for any given $\beta n$-sparse $\x$ with a fixed location of nonzero components and a priori known to be comprised of non-negative elements. Since the analysis will clearly be irrelevant with respect to what particular location of nonzero elements is chosen, we can for the simplicity of the exposition and without loss of generality assume that the components $\x_{1},\x_{2},\dots,\x_{n-k}$ of $\x$ are equal to zero and the components $\x_{n-k+1},\x_{n-k+2},\dots,\x_n$ of $\x$ are greater than or equal to zero. Under this assumption we have the following (see e.g. \cite{StojnicICASSP09}) ``signed" analogue to Theorem \ref{thm:thmgenweak}.
\begin{theorem}(Nonzero part of $\x$ has a fixed location; The signs of elements of $\x$ are a priori known)
Assume that an $m\times n$ matrix $A$ is given. Let $\x$
be a $k$-sparse vector whose nonzero components are known to be positive. Also let $\x_1=\x_2=\dots=\x_{n-k}=0.$
Further, assume that $\y=A\x$ and that $\w$ is
an $n\times 1$ vector. If
\begin{equation}
(\forall \w\in \textbf{R}^n | A\w=0, \w_{i}\geq 0,1\leq i\leq n-k) \quad  -\sum_{i=n-k+1}^n \w_i<\sum_{i=1}^{n-k}\w_{i}
\label{eq:thmeqgennonweak1}
\end{equation}
then the solution of (\ref{eq:l1signed}) is $\x$. Moreover, if
\begin{equation}
(\exists \w\in \textbf{R}^n | A\w=0, \w_{i}\geq 0,1\leq i\leq n-k) \quad  -\sum_{i=n-k+1}^n \w_i>\sum_{i=1}^{n-k}\w_{i}
\label{eq:thmeqgennonweak2}
\end{equation}
then there will be a $k$-sparse nonnegative $\x$ that satisfies (\ref{eq:system}) and is not the solution of (\ref{eq:l1signed}).
\label{thm:thmgennonweak}
\end{theorem}
 Below we probabilistically analyze validity of the null-space characterization given in the second part of Theorem \ref{thm:thmgennonweak}. Essentially, we will design a mechanism for computing upper bounds on $\betawnon$ (in fact, as it was the case in the previous section, since it will be slightly more convenient we will actually determine lower bounds on $\alpha$; that will of course again be conceptually the same as finding the upper-bounds on $\beta$).

We start by defining a quantity $\tauw$ which will be an analogue to $\tau$ from the previous section and will play one of the key roles below
\begin{eqnarray}
\tauw(A) =  \min & & (\sum_{i=1}^{n-k}\w_i+\sum_{i=n-k+1}^{n}\w_i) \nonumber \\
\mbox{subject to} & &  A\w=0\nonumber \\
& & \w_i\geq 0,1\leq i\leq n-k\nonumber \\
& & \|\w\|_2\leq 1.\label{eq:deftaunon}
\end{eqnarray}
We will continue to assume that the entries of $A$ are i.i.d. standard normal random variables. Similarly to what was established in (\ref{eq:mcrit}) we have that one can say that for any $\alpha$ and $\beta$ for which
\begin{equation}
\lim_{n\rightarrow\infty}P(\tauw(A)<0)=1, \label{eq:mcritnon}
\end{equation}
there is an a priori known to be nonnegative $k$-sparse $\x$ (from a set of $\x$'s that have given fixed location of non-zeros) that satisfies (\ref{eq:system}) and which (\ref{eq:l1signed}) with probability $1$ fails to find. For a fixed $\beta$ our goal will be to find the largest possible $\alpha$ for which (\ref{eq:mcritnon}) holds, i.e. for which (\ref{eq:l1signed}) fails with probability $1$.

As it was the case in the previous section, before going through the randomness of the problem and evaluation of $P(\tauw(A)<0)$, we will first try to provide an expression for $\tauw$ that is a bit more explicit than the one given by the optimization problem in (\ref{eq:deftaunon}). To facilitate following and exposition the rest of the analysis will parallel as much as possible what was presented in the previous section. To that end we start by writing the Lagrange dual of the optimization problem in (\ref{eq:deftaunon}) (as it was the case in the previous section, the strong duality trivially holds throughout the rest of this derivation). After regrouping the terms one has
\begin{eqnarray}
\tauw(A)= \max_{\lambda^{(1)},\nu,\gamma}\min_{\w} &  & \sum_{i=1}^{n-k}\w_i
+\sum_{i=n-k+1}^{n}\w_i-\sum_{i=1}^{n-k}\lambda^{(1)}\w_i+\nu^T A\w+\gamma\sum_{i=1}^{n}\w_i^2-\gamma\nonumber \\
& & \lambda_i^{(1)}\geq 0, 1\leq i\leq n-k\nonumber \\
& & \gamma\geq 0.
\label{eq:deftau3non}
\end{eqnarray}
At this point we proceed by solving the inner minimization over $\w$. To that end, let
\begin{equation}
f_1^+(\lambda^{(2)},\nu,\gamma,\w)=\sum_{i=n-k+1}^{n}\w_i+\sum_{i=1}^{n-k}(1-\lambda_i^{(1)})\w_i+\nu^T A\w +\gamma\sum_{i=1}^{n}\w_i^2-\gamma.\label{eq:deftau7non}
\end{equation}
Since $f_1^+(\cdot)$ is convex in $\w$ we simply find the optimal $\w$ by equaling the derivative of $f_1^+(\cdot)$ with respect to $\w$ to zero. We then have
\begin{eqnarray}
\frac{df_1^+(\lambda^{(1)},\nu,\gamma,\w)}{d\w_i} & = & (1-\lambda_i^{(1)})+\nu^T A_i +2\gamma\w_i, 1\leq i\leq n-k\nonumber \\
\frac{df_1^+(\lambda^{(1)},\nu,\gamma,\w)}{d\w_i} & = & 1+\nu^T A_i +2\gamma\w_i, n-k+1\leq i\leq n
\label{eq:deftau8non}
\end{eqnarray}
where as earlier $A_i$ is the $i$-th column of $A$. Let
\begin{equation}
\z^+=[(1-\lambda_1^{(1)}),(1-\lambda_2^{(1)}),\dots,(1-\lambda_{n-k}^{(1)}),1,1,\dots,1]^T.\label{eq:defznon}
\end{equation}
From (\ref{eq:deftau8}) one easily finds
\begin{eqnarray}
\w_{opt}^+=\frac{-\z^+-A^T\nu}{2\gamma}.\label{eq:solwnon}
\end{eqnarray}
Removing the inner minimization over $\w$ in (\ref{eq:deftau3non}) and recognizing the relation between $\z^+$ and $\lambda^{(1)}$ we have
\begin{eqnarray}
\tauw(A)= \max_{\z^+,\nu,\gamma}&  &
(\z-A^T \nu)^T\w_{opt}^+ +\gamma\|\w_{opt}^+\|_2^2-\gamma  \nonumber \\
\mbox{subject to}& & \z^+_i\leq 1, 1\leq i\leq n-k.\nonumber \\
& & \z^+_i=1, n-k+1\leq i\leq n\nonumber \\
& & \gamma\geq 0.
\label{eq:deftau9non}
\end{eqnarray}
Finally after plugging $\w_{opt}^+$ from (\ref{eq:solwnon}) in (\ref{eq:deftau9non}) we obtain
\begin{eqnarray}
\tauw(A)= \max_{\z^+,\nu,\gamma}&  &
-\frac{\|\z^+-A^T \nu\|_2^2}{4\gamma}-\gamma  \nonumber \\
\mbox{subject to}& & \z^+_i\leq 1, 1\leq i\leq n-k.\nonumber \\
& & \z^+_i=1, n-k+1\leq i\leq n\nonumber \\
& & \gamma\geq 0.
\label{eq:deftau10non}
\end{eqnarray}
The maximization over $\gamma$  is then trivial and one finally has
\begin{eqnarray}
\tauw(A)= \max_{\z^+,\nu}&  &
-\|\z^+-A^T \nu\|_2  \nonumber \\
\mbox{subject to}& & \z^+_i\leq 1, 1\leq i\leq n-k\nonumber \\
& & \z^+_i=1, n-k+1\leq i\leq n.
\label{eq:deftau11non}
\end{eqnarray}
or in a more convenient form
\begin{eqnarray}
\tauw(A)= -\min_{\z^+,\nu}&  &
\|\z^+-A^T \nu\|_2  \nonumber \\
\mbox{subject to}& & \z^+_i\leq 1, 1\leq i\leq n-k\nonumber \\
& & \z^+_i=1, n-k+1\leq i\leq n.
\label{eq:deftau12non}
\end{eqnarray}
At this point we are again almost ready to switch to the probabilistic aspect of the analysis. Again, the last piece of transformation is to rewrite (\ref{eq:deftau12non}) as
\begin{eqnarray}
\tauw(A)= -\min_{\z^+,\nu} \max_{\|\a\|_2=1} &  &
\a^T(\z^+-A^T \nu)  \nonumber \\
\mbox{subject to}& & \z^+_i\leq 1, 1\leq i \leq n-k\nonumber \\
& & \z^+_i=1, n-k+1\leq i\leq n.
\label{eq:deftau13non}
\end{eqnarray}
Now we are ready to invoke the results from Theorem \ref{thm:Gordonmesh1}. We do so through the following lemma which is a slightly modified Lemma \ref{eq:unsignedlemma} which itself is a slightly modified version of Lemma 3.1 from \cite{Gordon88}.
\begin{lemma}
Let $A$ be an $m\times n$ matrix with i.i.d. standard normal components. Let $\g$ and $\h$ be $n\times 1$ and $m\times 1$ vectors with i.i.d. standard normal components. Also, let $g$ be a standard normal random variable and let $Z^+$ be a set such that $Z^+=(\z^+\in R^n|\z^+_i=1,n-k+1\leq i\leq n \quad \mbox{and} \quad \z^+_i\leq 1, 1\leq i\leq n-k)$. Then
\begin{equation}
\hspace{-.6in}P(\min_{\z^+\in Z^+,\nu\in R^n\setminus 0}\max_{\|\a\|_2=1}(-\a^T A\nu +\|\nu\|_2 g-\zeta_{\a,\z^+,\nu})\geq 0)\geq P(\min_{\z^+\in Z^+,\nu\in R^n\setminus 0}\max_{\|\a\|_2=1}(\|\nu\|_2\sum_{i=1}^{n}\g_i\a_i+\sum_{i=1}^{m}\h_i\nu_i-\zeta_{\a,\z^+,\nu})\geq 0).\label{eq:problemmanon}
\end{equation}\label{eq:signedlemma}
\end{lemma}
\begin{proof}
The proof is again exactly the same as the one of Lemma 3.1 in \cite{Gordon88}. The only difference is that one should make identical copies over $\z$ of the processes $X_{x,y},Y_{x,y}$ defined in that proof. The rest of the proof remains unaltered.
\end{proof}
Let $\zeta_{\a,\z^+,\nu}=\epsilon_{5}^{(g)}\sqrt{n}\|\nu\|_2-\a^T\z^+$ with $\epsilon_{5}^{(g)}>0$ being an arbitrarily small constant independent of $n$. The left-hand side of the inequality in (\ref{eq:problemmanon}) is then the following probability of interest
\begin{equation*}
P(\min_{\z^+\in Z^+,\nu\in R^n\setminus 0}\max_{\|\a\|_2=1}(\|\nu\|_2\sum_{i=1}^{n}\g_i\a_i+\sum_{i=1}^{m}\h_i\nu_i-\epsilon_{5}^{(g)}\sqrt{n}\|\nu\|_2+\a^T\z^+)\geq 0).
\end{equation*}
One can then repeat the steps until (\ref{eq:probanal2}) from the previous section and obtain
\begin{multline}
P(\min_{\z^+\in Z^+,\nu\in R^n\setminus 0}\max_{\|\a\|_2=1}(\|\nu\|_2\sum_{i=1}^{n}\g_i\a_i+\sum_{i=1}^{m}\h_i\nu_i-\epsilon_{5}^{(g)}\sqrt{n}\|\nu\|_2+\a^T\z^+)\geq 0)\\=P(\min_{\z^+\in Z^+,\nu\in R^n\setminus 0}(\|\g+\frac{1}{\|\nu\|_2}\z^+\|_2)\geq \|\h\|_2+\epsilon_{5}^{(g)}\sqrt{n})\\\geq (1-e^{-\epsilon_{2}^{(m)} m})
P(\min_{\z^+\in Z^+,\nu\in R^n\setminus 0}(\|\g+\frac{1}{\|\nu\|_2}\z^+\|_2)\geq (1+\epsilon_{1}^{(m)})\sqrt{m}+\epsilon_{5}^{(g)}\sqrt{n})).\label{eq:probanal2non}
\end{multline}
Now, let $\bar{\g}^+=[\g_{(1)}^+,\g_{(2)}^+,\dots,\g_{(n-k)}^+,\g_{n-k+1},\g_{n-k+2},\dots,\g_n]^T$, where $[\g_{(1)}^+,\g_{(2)}^+,\dots,\g_{(n-k)}^+]$ are\\ $[\g_{1},\g_{2},\dots,\g_{n-k}]$ sorted in increasing order. Then clearly,
\begin{equation}
\min_{\z^+\in Z^+,\nu\in R^n\setminus 0}(\|\g+\frac{1}{\|\nu\|_2}\z^+\|_2)=\min_{\z^+\in |Z^+|,\nu\in R^n\setminus 0}(\|\bar{\g}^+-\frac{1}{\|\nu\|_2}\z^+\|_2)\label{eq:probanal3non}
\end{equation}
 Moreover, the optimization on the right-hand side of (\ref{eq:probanal2non}) is again structurally the same as the one in equation $(15)$ in \cite{StojnicCSetam09} (again, to be completely exact, it is the same as the nonnegative weak threshold equivalent to $(15)$). Essentially, the exact equivalence between these optimizations is achieved if in $(15)$ from \cite{StojnicCSetam09} $\tilde{\h}$ is replaced by $\bar{\g}^+$, $\nu$ is replaced by $\frac{1}{\|\nu\|_2}$, $\lambda$ is restricted to the lower $(n-k)$ components, $\z$ is replaced by $\z^+$, and one recalls that we earlier introduced simplification $\z^+_i=1-\lambda^{(1)},1\leq i\leq n-k$,$\lambda^{(1)}\geq 0$ (that would in essence be the nonnegative weak threshold equivalent to $(15)$; similarly to what we have mentioned in the previous section, the nonnegative weak threshold equivalent to $(15)$ was not explicitly written anywhere in \cite{StojnicCSetam09} but is rather obvious; of course, as mentioned earlier, in \cite{StojnicCSetam09} we made a nonnegative ``weak" equivalence to $(29)$ instead). With these replacements one can then use the machinery of \cite{StojnicCSetam09} to establish
\begin{equation}
\min_{\z^+\in Z,\nu\in R^n\setminus 0}(\|\bar{\g}^+-\frac{1}{\|\nu\|_2}\z^+\|_2)=\sqrt{\sum_{i=\cweaknon+1}^n(\bar{\g}_i^+)^2-\frac{(((\bar{\g}^+)
^T\z)-\sum_{i=1}^{\cweaknon}\bar{\g}_i^+)^2}{n-\cweaknon
}}=f_{\g^+}(\cweaknon)\label{eq:probanal4non}
\end{equation}
where $\cweaknon$ is the solution of
\begin{equation}
\frac{(\bar{\g}
^T\z)-\sum_{i=1}^{\cweaknon}\bar{\g}_i}{n-\cweaknon
}=\bar{\g}_{\cweaknon}.\label{eq:defcwnon}
\end{equation}
We recall again, that the key point to success of our method is that the derivation of \cite{StojnicCSetam09}
establishes equalities in (\ref{eq:probanal4non}) and (\ref{eq:defcwnon}). At this point we have established the core of our upper-bounding arguments for the nonnegative case. The rest would be just a repetition of the derivations done in the previous section that would make everything precise. After replacing $\bar{\g}$ by $\bar{\g}^+$ and repeating literally every step of the derivation in the previous section from (\ref{eq:probanal44}) until (\ref{eq:probanal7}) one arrives to a ``nonnegative" equivalent to (\ref{eq:probanal7})
\begin{equation}
\hspace{-.3in}P(\min_{\z^+\in Z^+,\nu\in R^n\setminus 0}\max_{\|\a\|_2=1}(-\a^T A\nu+\a^T\z^++\|\nu\|_2(g-\epsilon_{5}^{(g)}\sqrt{n}))\geq 0)\geq (1-e^{-\epsilon_{2}^{(m)}m_w})
(1-e^{-\epsilon_{2}^{(g)} n})(1-e^{-\epsilon_{4}^{(g)} n})(1-e^{-\epsilon_{3}^{(c)}n}).\label{eq:probanal7non}
\end{equation}
Repeating then the last piece of argument related to $P(g\leq\epsilon_{5}^{(g)}\sqrt{n})\geq 1-e^{-\epsilon_{6}^{(g)} n}$ one then arrives at
\begin{equation*}
P(-\tauw(A)> 0)\geq (1-e^{-\epsilon_{2}^{(m)}m_w})
(1-e^{-\epsilon_{2}^{(g)} n})(1-e^{-\epsilon_{4}^{(g)} n})(1-e^{-\epsilon_{6}^{(g)} n})(1-e^{-\epsilon_{3}^{(c)}n}),
\end{equation*}
and ultimately at
\begin{equation}
\lim_{n\rightarrow\infty}P(\tauw(A)< 0)=\lim_{n\rightarrow\infty} (1-e^{-\epsilon_{2}^{(m)}m_w})
(1-e^{-\epsilon_{2}^{(g)} n})(1-e^{-\epsilon_{4}^{(g)} n})(1-e^{-\epsilon_{6}^{(g)} n})(1-e^{-\epsilon_{3}^{(c)}n})=1,\label{eq:finaltaunon}
\end{equation}
which is what we established as a goal in (\ref{eq:mcritnon}).
We summarize the results in the following theorem.

\begin{theorem}(Exact weak threshold --- signed $\x$)
Let $A$ be an $m\times n$ matrix in (\ref{eq:system})
with i.i.d. standard normal components. Let
the unknown $\x$ in (\ref{eq:system}) be $k$-sparse and let it be a priori known that its nonzero components are positive. Further, let the location of the nonzero elements of $\x$ be arbitrarily chosen but fixed.
Let $k,m,n$ be large
and let $\alpha=\frac{m}{n}$ and $\beta_w^+=\frac{k}{n}$ be constants
independent of $m$ and $n$. Let $\erfinv$ be the inverse of the standard error function associated with zero-mean unit variance Gaussian random variable.  Further,
let all $\epsilon$'s below be arbitrarily small constants.
\begin{enumerate}
\item Let $\htheta_w^+$, ($\beta_w^+\leq \htheta_w^+\leq 1$) be the solution of
\begin{equation}
(1-\epsilon_1^{(c)})(1-\betawnon)\frac{\sqrt{\frac{1}{2\pi}}e^{-(\erfinv(2\frac{1-\thetawnon}{1-\betawnon}-1))^2}}{\thetawnon}-\sqrt{2}\erfinv ((2\frac{(1+\epsilon_1^{(c)})(1-\thetawnon)}{1-\betawnon}-1))=0.\label{eq:thmweaknontheta1}
\end{equation}
If $\alpha$ and $\beta_w^+$ further satisfy
\begin{equation}
\hspace{-.7in}\alpha>\frac{1-\betawnon}{\sqrt{2\pi}}\left (\frac{\sqrt{2(\erfinv(2\frac{1-\hthetawnon}{1-\betawnon}-1))^2}}{e^{(\erfinv(2\frac{1-\hthetawnon}{1-\betawnon}-1))^2}}\right )+\hthetawnon
-\frac{\left ((1-\betawnon)\sqrt{\frac{1}{2\pi}}e^{-(\erfinv(2\frac{1-\hthetawnon}{1-\betawnon}-1))^2}\right )^2}{\hthetawnon}\label{eq:thmweakalphanon}
\end{equation}
then with overwhelming probability the solution of (\ref{eq:l1signed}) is the positive $k$-sparse $\x$ from (\ref{eq:system}).
\item Let $\htheta_w^+$, ($\beta_w^+\leq \htheta_w^+\leq 1$) be the solution of
\begin{equation}
(1+\epsilon_2^{(c)})(1-\betawnon)\frac{\sqrt{\frac{1}{2\pi}}e^{-(\erfinv(2\frac{1-\thetawnon}{1-\betawnon}-1))^2}}{\thetawnon}-\sqrt{2}\erfinv ((2\frac{(1-\epsilon_2^{(c)})(1-\thetawnon)}{1-\betawnon}-1))=0.\label{eq:thmweaknontheta2}
\end{equation}
If on the other hand $\alpha$ and $\beta_w^+$ satisfy
\begin{multline}
\hspace{-1in}\alpha<\frac{1}{(1+\epsilon_{1}^{(m)})^2}\left ((1-\epsilon_{1}^{(g)})(\htheta_w^++\frac{(1-\beta_w^+)}{\sqrt{2\pi}} \frac{\sqrt{2(\erfinv(2\frac{1-\htheta_w^+}{1-\beta_w^+}-1))^2}}{e^{(\erfinv(2\frac{1-\htheta_w^+}{1-\beta_w^+}-1))^2}})
-\frac{\left ((1-\beta_w^+)\sqrt{\frac{1}{2\pi}}e^{-(\erfinv(2\frac{1-\hat{\theta}_w^+}{1-\beta_w^+}-1))^2}\right )^2}{\hat{\theta}_w^+(1+\epsilon_{3}^{(g)})^{-2}}\right )\label{eq:thmweaknonalpha2}
\end{multline}
then with overwhelming probability there will be a positive $k$-sparse $\x$ (from a set of $\x$'s with fixed locations of nonzero components) that satisfies (\ref{eq:system}) and is \textbf{not} the solution of (\ref{eq:l1signed}).
\end{enumerate}
\label{thm:thmweaknonthr}
\end{theorem}
\begin{proof}
The first part was established in \cite{StojnicCSetam09}. The second part follows from the previous discussion combining (\ref{eq:thmeqgennonweak2}), (\ref{eq:deftaunon}), (\ref{eq:mcritnon}), (\ref{eq:probanal4non}), (\ref{eq:defmmax}),  and (\ref{eq:finaltaunon}) and the corresponding derivation from the previous section.
\end{proof}

\noindent As it was the case in the previous section, the previous theorem insists on precision and involves ``epsilon" type of characterization. However, one can again do what we, in Section \ref{sec:unsigned}, referred to as the ``deepsilonification" and obtain a way more convenient characterization. After removing all $\epsilon$'s one in a more informal language then has.

\noindent Assume the setup of the above theorem. Let $\alpha_w^+$ and $\beta_w^+$ satisfy the following:

\noindent \underline{\underline{\textbf{Fundamental characterization of the $\ell_1$ performance ($\x$ in (\ref{eq:system}) a priori known to be positive):}}}

\begin{center}
\shadowbox{$
(1-\beta_w^+)\frac{\sqrt{\frac{1}{2\pi}}e^{-(\erfinv(2\frac{1-\alpha_w^+}{1-\beta_w^+}-1))^2}}{\alpha_w^+}-\sqrt{2}\erfinv (2\frac{1-\alpha_w^+}{1-\beta_w^+}-1)=0.
$}
-\vspace{-.5in}\begin{equation}
\label{eq:thmweaknontheta2}
\end{equation}
\end{center}

Then:
\begin{enumerate}
\item If $\alpha>\alpha_w^+$ then with overwhelming probability the solution of (\ref{eq:l1signed}) is the a priori known to be positive $k$-sparse $\x$ from (\ref{eq:system}).
\item If $\alpha<\alpha_w^+$ then with overwhelming probability there will be an a priori known to be positive $k$-sparse $\x$ (from a set of $\x$'s with fixed locations of nonzero components) that satisfies (\ref{eq:system}) and is \textbf{not} the solution of (\ref{eq:l1signed}).
    \end{enumerate}

As stated above equation (\ref{eq:thmweaknontheta2}) is the fundamental characterization of the $\ell_1$ performance when applied to recovery of vectors $\x$ that are a priori known to have positive (in general of any sign) nonzero components. Numerical values of the ``nonnegative" weak threshold obtained using (\ref{eq:thmweaknontheta2}) were presented in \cite{StojnicCSetam09}. As it was demonstrated there, the lower bounds on the thresholds were in an excellent numerical agreement with the optimal thresholds computed in \cite{DT,DonohoSigned}. Theorem  \ref{thm:thmweaknonthr} establishes that the lower bounds computed in \cite{StojnicCSetam09} (essentially those one can compute from (\ref{eq:thmweaknontheta2})) are actually the upper bounds as well and as such are the exact values of the weak thresholds. Moreover, as it was the case for the general vectors $\x$ from the previous section, in a companion paper \cite{StojnicEquiv10} we established a qualitative equivalence of the characterization given in (\ref{eq:thmweaknontheta2}) and the results obtained in \cite{DT}. Again in a rather fascinating way the axiomatic system of mathematics works so well that two seemingly different approaches, the geometric one from \cite{DT} and the purely probabilistic one from \cite{StojnicCSetam09}, result in exactly the same optimal characterization of the performance of $\ell_1$-optimization.

We one more time, out of respect for an incredible effort that was put forth to characterize the $\ell_1$ performance in  \cite{DonohoSigned,DT,StojnicCSetam09,StojnicEquiv10} present in Figure \ref{fig:weakthr} the plot obtained based on the ultimate ``nonnegative" characterization (\ref{eq:thmweaknontheta2}).
\begin{figure}[htb]
\centering
\centerline{\epsfig{figure=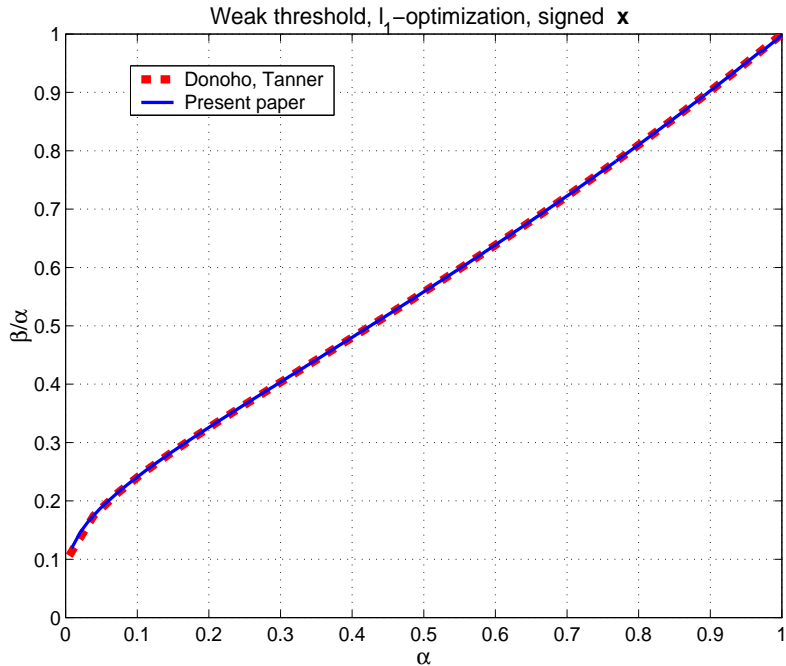,width=10.5cm,height=9cm}}
\vspace{-0.2in} \caption{\emph{Weak} threshold, $\ell_1$-optimization --- ultimate performance}
\label{fig:weakthr}
\end{figure}

\section{Discussion}
\label{sec:discuss}

In this paper we considered under-determined linear systems of equations with sparse solutions.
We looked from a theoretical point of view at a classical polynomial-time
$\ell_1$-optimization algorithm. Under the assumption that the system matrix $A$ has i.i.d. standard normal components,
we derived upper bounds on the values of the recoverable weak
thresholds in the so-called linear regime, i.e. in the regime when
the recoverable sparsity is proportional to the length of the
unknown vector. Obtained upper bounds match the corresponding lower bounds we found through a framework designed in \cite{StojnicCSetam09}. Combination of the mechanism from \cite{StojnicCSetam09} and the one that we presented in this paper is then enough to provide an explicit ultimate characterization of the success of $\ell_1$ optimization when applied in solving under-determined systems of linear equations with sparse solutions.

Further developments are pretty much then unlimited (though, of course their scientific value will never match the one of the results presented in \cite{StojnicCSetam09} and here). Namely, we hardly ever encountered a ``sparse recovery" type of the problem where the lower-bounding technique from \cite{StojnicCSetam09} was not exact. The mechanism that we designed in this paper then helps to make all the success of \cite{StojnicCSetam09} ultimate, i.e. it helps to prove what ultimately can be proved for this type of optimization problems.

Various specific problems that have been of interest in a broad scientific literature developed over the last few years, like quantifying the performance of $\ell_1$ type of optimization problems in solving systems with special structure of the solution vector (block-sparse, binary, box-constrained, low-rank matrix, partially known locations of nonzero components, just to name a few), systems with non-exact (noisy) solution vectors and/or equations can then easily be handled. In a few forthcoming companion papers we will present some of these applications. However, as it will be clear when these results appear, each of them will require some work to put the mechanism forth but in essence they all will be fairly simple extensions of what we presented in \cite{StojnicCSetam09} and here. The heart of it all will really be the lower-bounding mechanism designed in \cite{StojnicCSetam09} and the complementary upper-bounding mechanism designed in this paper and how the two ultimately meet in a somewhat magical way.

\begin{singlespace}
\bibliographystyle{plain}
\bibliography{UpperBound}
\end{singlespace}

\end{document}